\documentclass[
    aps, twocolumn,
    nobibnotes,
    superscriptaddress,10pt]{revtex4-2}

\usepackage{graphicx}
\usepackage{float}

\makeatletter
\let\newfloat\newfloat@ltx
\makeatother

\usepackage{dcolumn}
\usepackage{bm}
\usepackage[breakable,skins]{tcolorbox}
\usepackage{physics}
\usepackage{amsthm}
\usepackage{amssymb}
\usepackage{hyperref}
\usepackage{algorithm}
\usepackage[noend]{algpseudocode}

\algblockdefx{MRepeat}{EndRepeat}{\textbf{repeat} }{}

\usepackage{quantikz}

\usepackage{xcolor}

\newcommand{\fref}[1]{Fig.~\ref{#1}}
\newcommand{\past}[1]{\overleftarrow{#1}}
\newcommand{\future}[1]{\overrightarrow{#1}}
\renewcommand{\eqref}[1]{Eq.~(\ref{#1})}

\newcommand{\prob}{\mathbb{P}}

\DeclareMathAlphabet{\bboldbb}{U}{bbold}{m}{n}
\newcommand{\one}{\bboldbb{1}}
\newcommand{\chU}{\mathcal{U}}
\newcommand{\chW}{\mathcal{W}}

\newcommand{\inlineheading}[1]{\textbf{#1~-- }\ignorespaces}

\theoremstyle{plain}
\newtheorem{thm}{Theorem}

\newtheorem{prop}[thm]{Proposition}

\theoremstyle{definition}

\theoremstyle{remark}

\begin{document}
\title{Learning Quantum-Samplers for Stochastic Processes with Quantum Sequence Models}

\author{Ximing Wang}
\email{canoming.sktt@gmail.com}
\affiliation{
Nanyang Quantum Hub, School of Physical and Mathematical Sciences, Nanyang Technological University, Singapore
}
\affiliation{
Centre for Quantum Technologies, Nanyang Technological University, Singapore
}
\author{Chengran Yang}
\affiliation{
Nanyang Quantum Hub, School of Physical and Mathematical Sciences, Nanyang Technological University, Singapore
}
\affiliation{
Centre for Quantum Technologies, Nanyang Technological University, Singapore
}

\author{Chidambaram Aditya Somasundaram}
\affiliation{
Nanyang Quantum Hub, School of Physical and Mathematical Sciences, Nanyang Technological University, Singapore
}
\affiliation{
Centre for Quantum Technologies, Nanyang Technological University, Singapore
}

\author{Jayne Thompson}
\email{jayne.thompson@ntu.edu.sg}
\affiliation{
College of Computing and Data Sciences Nanyang Technological University, Singapore
}
\affiliation{
School of Physical and Mathematical Sciences, Nanyang Technological University, Singapore
}
\affiliation{
Centre for Quantum Technologies, Nanyang Technological University, Singapore
}

\author{Mile Gu}
\email{gumile@ntu.edu.sg}
\affiliation{
Nanyang Quantum Hub, School of Physical and Mathematical Sciences, Nanyang Technological University, Singapore
}
\affiliation{
Centre for Quantum Technologies, Nanyang Technological University, Singapore
}
\affiliation{MajuLab, CNRS-UNS-NUS-NTU International Joint Research Unit, UMI 3654, Singapore 117543, Singapore}
\date{\today}

\begin{abstract}Quantum circuits that generate coherent superpositions of stochastic processes are key to many downstream quantum-accelerated tasks, such as risk analysis, importance sampling, and DNA sequencing. However, traditional methods for designing such circuits from data face immense challenges, given the exponential growth in the size of the associated probability vectors as the desired simulation time horizon increases. Here, we introduce quantum sequence models that leverage a recurrent quantum circuit structure to generate coherent superpositions with circuit complexity that grows linearly with the desired time horizon; together with a recurrent variant of the parameter-shift rule, we train these models from observational data. When benchmarked against baseline quantum Born machines, our constructions exhibit orders-of-magnitude improvements in model accuracy in data-sparse regimes.
\end{abstract}

\maketitle

\noindent\inlineheading{Introduction} Monte Carlo simulation is a key tool in forecasting future behavior in stochastic processes - allowing us to sample future trajectories and thus infer the likelihood of various future outcomes~\cite{rubinstein2016,gilks1995,hastings1970}. In such tasks, quantum computing devices offer a distinct advantage. Rather than producing individual samples, they enable \emph{quantum samplers} that can coherently generate \emph{quantum samples (q-samples)} - quantum states whose probability amplitudes encode all possible sequences in superposition weighted by their corresponding probabilities. Provided we know a quantum circuit realization of such \emph{quantum samplers}, their use as a subroutine in various quantum algorithms enables powerful downstream applications, including quantum-accelerated risk analysis~\cite{woerner2019}, importance sampling~\cite{blank2021}, and motif detection in genomic data~\cite{montanaro2017}.

Such algorithms typically assume that a quantum circuit to generate such q-samples is readily available. In real-world settings, however, we typically begin with a finite string of observational data drawn from some underlying unknown process. Thus, the capacity to leverage the above advantages depends crucially on our capability to infer quantum samplers from this data. In general, this is non-trivial. The probability distributions of a stochastic process over $L$ time-steps are described by a probability vector of size exponentially in $L$. Thus, systematic methods to design a quantum sampler for $P$ exhibit gate complexity with similar exponential growth~\cite{binder2018, yang2025}. Meanwhile, standard variational circuits to learn this q-sample would suffer from exponentially suppressed gradients~\cite{mcclean2018, wang2021a}.

\begin{figure}[htp]
    \centering
    \includegraphics[width=0.98\linewidth]{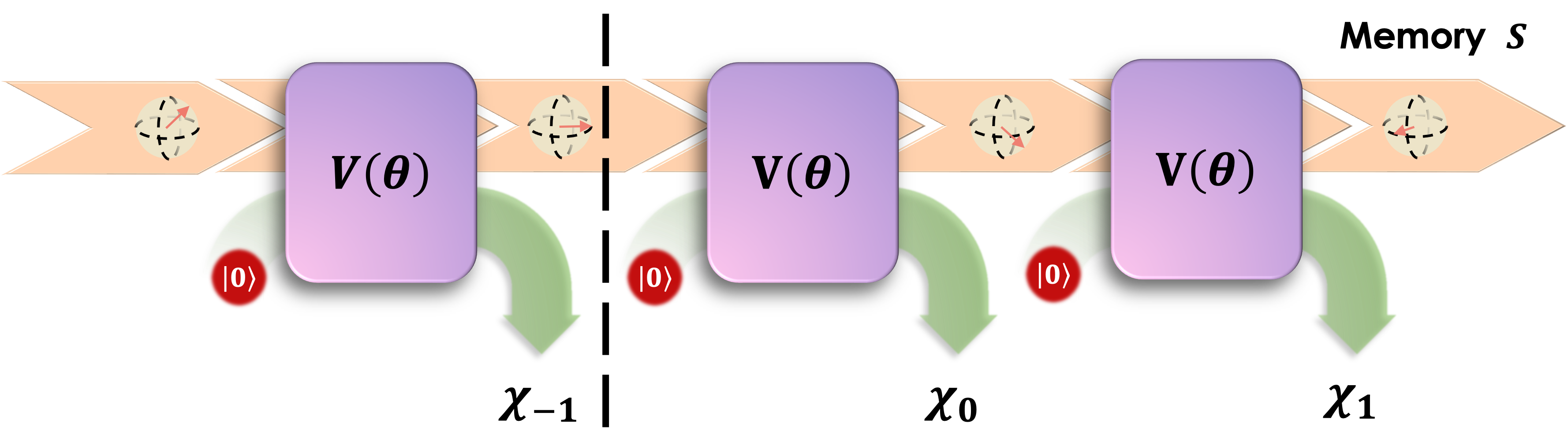}
    \caption[\textbf{Quantum Sequence Models.}]{\textbf{Quantum Sequence Models.} \label{fig:QModel} A quantum sequence model $\mathcal{Q}_{\theta}$ consists of a memory $S$ and a sequence of output registers $\qty{\chi_t}$. Its dynamics is governed by a unitary $V(\vb*{\theta})$, which interacts $S$ with $\chi_t$ at each time-step $t$. Here, the $\vb*{\theta}$-dependence represents some parameterization of such sequence models. We say that $\mathcal{Q}_{\theta}$ is a quantum-sampler for a stochastic process $Q_{\theta}(\past{X},\future{X})$ if its resulting measurement statistics, when each $\chi_t$ is measured, are governed by $Q$. Our goal is to learn $\vb*{\theta}$ such that $Q$ closely approximates a given target stochastic process $P$.}
\end{figure}

Under this context, we introduce tools to train \emph{quantum sequence models} - a recurrent quantum circuit in which a single parameterized quantum operation is applied repeatedly to an internal quantum memory and a new output register at each time-step (illustrated in \fref{fig:QModel}). When given a string of data sampled from a stochastic process $P$, the resulting recurrent quantum circuit can prepare the output registers in the $L$-step q-sample in gate complexity linear in $L$. We compare our resulting training methods to Born machines~\cite{gili2023, liu2018} - conventional variational quantum circuits designed to generate q-samples - and demonstrate that the recurrent structure has a significant learning advantage: it delivers q-samplers with orders-of-magnitude improvements in statistical accuracy on limited observational data.

\noindent\inlineheading{Stochastic Modeling and Q-Sampling} We begin by assuming that our system of interest is governed by some underlying stochastic process $P$ with random variable $X_t$ that emits a sequence of outputs $x_t$ at each timestep $t\in\mathbb{Z}$~\cite{khintchine1934}. Taking $t = 0$ as the present, the statistics of such systems can be described by the probability function $\smash{P(\past{X}, \future{X})}$ that correlates $\past{X} := \ldots X_{-2},X_{-1}$ with future $\future{X} = X_0, X_1, X_2,\ldots$. Here we assume $P$ is stationary, such that $P(\past{X}, \future{X})$ has no explicit dependence on $t$. We thus take $t = 0$ as the present without loss of generality.

A model of $P$ is \emph{statistically faithful} if it replicates the behavior of the statistical process. That is, given it has emitted past $\past{x}$, its corresponding future statistics are governed by $P(\future{X}|\past{x})$. In our context, we wish to build statistically faithful \emph{quantum samplers}. Consider an output register $\chi_t$ for each time-step $t$ and some ancillary memory register $S$, all expressed as qubits initially in state $\ket{0}$. The faithful quantum-sampler is then a unitary circuit that prepares a tri-partite quantum state $\ket{\Phi}_{\past{\chi},\future{\chi},S}$ such that measurement of $\past{\chi}$ in the computational basis with outcomes $\past{x}$ collapses $\past{\chi}$ to \emph{conditional q-samples}
$\ket{\Phi}^{\past{x}}_{\future{\chi},S}$ such that measurements of $\future{\chi}$ in the computational basis yields outcomes $\future{x}$ with probability
$P(\future{X} = \future{x}|\past{x})$.

In practice, finite data imply that any models we learn will never exactly exhibit the desired distribution $P(\future{X}| \past{x})$.  Distortion of the statistical outputs of the model is unavoidable. A key measure of distortion is the KL-divergence~\footnote{The KL-divergence of a probability distribution $P$ from a probability distribution $Q$ is defined as $D_{KL}(P,Q)={\Sigma_{x}} P(x)\log\frac{P(x)}{Q(x)}$.}, given its clear operational meaning across operational and physical contexts. Consider a candidate quantum sampler $\mathcal{Q}$ that exhibits outcomes $\future{x}$ with probability $\smash{Q(\future{X}|\past{x})}$. Let $\smash{\vb{p}_L^{\past{x}} = P(\future{X}_{L}|\past{x})}$, $\smash{\vb{q}_L^{\past{x}} = Q(\future{X}_{L}|\past{x})}$, with $\future{X}_{L} = X_0 X_1 \ldots X_{L-1}$. We then benchmark the quantum sampler by the KL-divergence rate~\cite{rached2004}
\begin{equation}\label{eq:divergeRate}
    D_e(P\Vert Q) := \lim_{L\to \infty}  \frac{1}{L}\sum_{\past{x}} P(\past{x}) D_{\mathrm{KL}}(\vb{p}_L^{\past{x}}\Vert \vb{q}_L^{\past{x}}).
\end{equation}
representing how quickly the future behavior of our q-sample deviates from that of the target distribution (when averaged across all pasts). Our goal is to learn quantum samplers where this rate is minimal.

\noindent\inlineheading{Finite Time Horizons} The above framework deals with the ideal setting where we have access to the infinite past and the infinite future, and can benchmark our models across this bi-infinite time-interval. In practice, it is usually impossible to generate a bi-infinite quantum sample, nor evaluate $D_e(P||\mathcal{Q})$ in the infinite limit. Instead, we likely care about causal dependencies up to $M$ time-steps into the past and about statistical faithfulness of predictions $L$ time-steps into the future. Let $\past{x}_M = x_{-M:0}$, this motivates the finite-horizon KL-divergence rate

$$D_{L,M}(P,Q) = \frac{1}{L}\sum_{\past{x}_M} P(\past{x}_M) D_{\mathrm{KL}}(\vb{p}_L^{\past{x}_M}\Vert \vb{q}_L^{\past{x}_M})$$

\noindent which evaluates the distortion of a candidate model based on how well its conditional distribution $\vb{q}_L^{\past{x}_M} = Q(\future{X}_L|\past{x}_M)$
matches that of the true distribution $\vb{p}_L^{\past{x}_M} = P(\future{X}_L|\past{x}_M)$ over time interval $-M \leq t < L$. In analogy, we define a q-sampler $\mathcal{Q}_{L,M}$ over this time interval as one that synthesizes the corresponding quantum sample $\ket{\Phi}_{\past{\chi}_M, \future{\chi}_L,S}$. When the context is clear, we drop the subscripts $L$ and $M$ for succinctness.

\noindent\inlineheading{Quantum Sequence Models}
We consider quantum sequence models as a potential model class to train such q-samplers. Such models are based on a recurrent quantum circuit $\mathcal{Q}_{\vb*{\theta}}$ (see~\fref{fig:QModel}). The circuit consists of (a) a sequence of output registers $\chi_t$ (b) an ancillary memory system $S$ that sequentially interacts with each output register via a unitary interaction parametrized by the vector of parameters $\vb*{\theta}$.
At each timestep $t$ from $t = - M$ to $t = L-1$, we interact the output register $\chi_t$ with $S$ via a unitary $V_{\vb{\theta}}$.  This then generates a quantum state $\ket{\Phi}_{\past{\chi}_M, \future{\chi}_L,S}$ across output registers and the ancillary memory system.

The recurrent structure distinguishes our model from the conventional variational settings in that
(i) As the same unitary $V_{\vb{\theta}}$ is applied at each time-step, the number of trainable parameters remains constant regardless of the length of the output sequence.
(ii) When the memory system interacts with an output system, the input state of the output system is always fixed to $\ket{0}$.
Applying $V_{\vb{\theta}}$ recurrently simplifies the circuit structure, where the number of parameters will not scale with the length of the output sequence.

In supplementary materials $D$, we outline a systematic means to parametrize $V_{\vb{\theta}}$. Here we illustrate that as initial register states are fixed at $\ket{0}$, we can further reduce the number of parametrized gates in the unitary $V_{\vb{\theta}}$ without losing expressivity. For example, for 2-qubit $V_{\vb{\theta}}$ (i.e., binary outputs with a single qubit memory), a full parameterization of $V_{\vb{\theta}}$ involves $15$ parameters~\cite{mottonen2004}. In the context of sequence models, however, only $8$ parameters are sufficient.

\noindent\inlineheading{Model Learning} Our model learning process begins with a data string $\mathcal{T}$ of length $T$ drawn from the target stochastic process $P$. Through binning and frequency counting, this data sequence then defines the empirical conditional distribution $\hat{P}(\future{X}_L|\past{x}_M)$ that converges to the true distribution $P(\future{X}_L|\past{x}_M)$ in the limit of infinite training data, i.e., as $T \to \infty$ (see Supplementary Materials section A.). Since we do not have access to the true distribution, the training phase of a learning algorithm uses $\hat{P}$ as a proxy for $P$. That is, from $\mathcal{T}$, we first train a quantum sampler
$\mathcal{Q}$ that generates samples of $Q(\future{X}_L|\past{x}_M)$. We can benchmark its performance based on how well it aligns with $\hat{P}(\future{X}_L|\past{x}_M)$. This is then given by the empirical distortion $\hat{D}(\mathcal{Q}) = D_{L,M}(\hat{P}\Vert Q)$. The difference between this and the true distortion $
\Delta D(\mathcal{Q}) = D(\mathcal{Q}) - \hat{D}(\mathcal{Q})$
is often referred to as the \emph{generalization gap}~\cite{caro2022,banchi2021}; representing the performance difference of a model between training data and unseen data.

Here, we begin with a quantum model class $\{Q_\theta\}$ parameterized by $\theta$. Our first step in training is to outline a cost function $f(\mathcal{Q}_{\vb{\theta}})$ that (i) is a faithful proxy to $\hat{D}(\mathcal{Q}_\theta)$ such that they are minimal on the same values $\theta$ and (ii) whose gradient with respect to $\theta$ can be efficiently computed by a quantum device that can execute $\mathcal{Q}_{\theta}$. The parameter shift rule represents industry-standard in variational circuits, where the trainable parameters are the angles in parameterized Pauli rotation gates~\cite{schuld2019} or time in Hamiltonian evolutions~\cite{banchi2021a}. However, to apply it to our model, we need to begin by generalizing this rule to accommodate the recurrent structure of quantum sequence models.

Specifically, consider a parameterized recurrent quantum circuit $\mathcal{Q}_{\vb*{\theta}}$ with repeated interactions described by $V(\vb*{\theta})$. In this context, the unitary $V(\vb*{\theta})$ is repeated $M+L$ times sequentially as in FIG.~\ref{fig:QModel}. As in the standard parameter shift rule, we assume that each $\theta \in \vb*{\theta}$ parametrizes $U$ by specifying the rotation angles $e^{i \frac{\theta}{2} \sigma}$, where $\sigma$ is any Pauli operator.
Without loss of generality, we focus here on a single parameter (see Supplementary Materials section B). Let $\ket{\Phi}_{\past{\chi}_M, \future{\chi}_L,S}$ be resulting quantum sample generated by the recurrent circuit, and let $Q(x_{-M:L})$ be its associated probability of outputting $x_{-M:L}$ when output registers are appropriately measured. Consider now a cost function $\lambda(x_{-M:L})$ that associates each visible sequence with a cost. We can then define the associated loss function as

\begin{equation}f(\mathcal{Q}_\theta) = \sum_{x_{-M:L}} \lambda(x_{-M:L}) Q(x_{-M:L}) \label{eq:definition_weighted_mean}\end{equation}

\noindent representing the expected cost under the distribution $Q$. In Algorithm~\ref{alg:ParameterShift}, we outline a system to efficiently estimate $\frac{\delta}{\delta \theta} f(\mathcal{Q}_{\theta})$ for any such $f$. The key idea is to introduce an alternative quantum circuit $\mathcal{Q}^{s,i}_{\theta}$ which replaces the $i$th application of $V(\theta)$ replaced by $V(\theta + s\frac{\pi}{2})$ for $s \in \{-1,+1\}$. The gradient can then be suitably estimated to multiplicative accuracy $\epsilon$ using $O\qty(\epsilon^{-2})$ calls to $\mathcal{Q}^{s,i}_{\theta}$ over a uniform sampling of $i$ and $s$. Our next theorem confirms that this is a valid and unbiased estimator.

\begin{thm}[Recurrent Parameter Shift]
    \label{thm:ParameterShift} Let $f$ be a linear expectation functional with associated cost functions $\lambda(x_{-M:L})$, the output $G$ of Algorithm~\ref{alg:ParameterShift} is an unbiased estimator of the gradient $\pdv{\theta}f(\mathcal{Q}_\theta)$.
\end{thm}

\begin{algorithm}[tbh]
\caption{Parameter-Shift Rule (Recurrent)}\label{alg:ParameterShift}
\begin{algorithmic}[1]
    \State \textbf{Input:} A recurrent circuit $\mathcal{Q}_{-M,L}(\theta)$ with $L+M$ repetitions of $V(\theta)$ as in Fig. \ref{fig:QModel};
    \State \textbf{Input:} A set of cost values $\lambda(x_{-M:L})$ as in \eqref{eq:definition_weighted_mean};
    \State \textbf{Input:} Number of samples $m$;
    \State \textbf{Initialize} empty list $\mathbf{g}$;
    \For{$k = 1\ldots m$}
        \State Uniformly sample a sign $s \in \qty{-1, +1}$ and an index $i \in [-M, L-1]$;
        \State Apply a shift to the parameter of the $i$-th instance of $V(\theta)$ in $\mathcal{Q}_{\theta}$: $\theta \to (\theta + s\frac{\pi}{2})$;
        \State Prepare the quantum sample $\ket{\Phi}_{\past{\chi}_M, \future{\chi}_L,S}^{s,i}$ using this modified circuit;
        \State Measure $\ket{\Phi}_{\past{\chi}_M, \future{\chi}_L,S}^{s,i}$ in the computational basis to obtain the string $x_{-M:L}$;
        \State Append $g_k = (M+L) s \lambda(x_{-M:L})$ to $\bf{g}$;
    \EndFor
    \State \textbf{End for}
    \State Set $G = \frac{1}{m}\sum_{j=1}^m g_j$ as the mean of $\mathbf{g}$;
    \State \textbf{Output:} $G$;
\end{algorithmic}
\end{algorithm}

\noindent\inlineheading{A Proxy for Distortion} The above algorithm leads to estimates of $D(\mathcal{Q})$, that have significant sensitivity to shot noise. This motivates us to introduce the co-emission distortion rate as a more computationally efficient proxy for distortion during training

\begin{equation}\label{eq:CoemitDiverge}
    \begin{aligned}
    D_A^L(\mathcal{Q})
    = -\frac{1}{L} \sum_{\past{x}_M} P(\past{x}_M)\log \frac{\vb{p}^{\past{x}_M}_L \cdot \vb{q}^{\past{x}_M}_L}
    {\norm{\vb{p}^{\past{x}_M}_L}\norm{\vb{q}^{\past{x}_M}_L}}.
    \end{aligned}
\end{equation}
where $\norm{\cdot}$ is the Euclidean norm and $\cdot$ is the standard vector dot product. The co-emission distortion holds the operational significance in context ranging from data-mining to natural language processing~\cite{singhal2001,church2017,reimers2019,wang2008,liang2019}, and is a faithful measure of statistical distortion~\cite{yang2020}.  In supplementary materials C, we show that the gradients of each of the $\theta$ dependent terms, $u(\theta) = \hat{\vb{p}}^{\past{x}_M}_L\cdot \hat{\vb{q}}^{\past{x}_M}$ and $v(\theta) = \norm{\hat{\vb{q}}^{\past{x}_M}_L}$, are efficiently measurable on a quantum computer.

This is done by observing that both $u(\theta)$ and $v(\theta)$ can be written in the form of a linear expectation functional (For example, to evaluate $\hat{\vb{p}}^{\past{x}_M}_L\cdot\hat{\vb{q}}^{\past{x}_M}_L(\vb*{\theta})$, we take $\lambda(x_{-M:L})$ to be $\hat{\vb{p}}^{\past{x}_M}_L$). Theorem~\ref{thm:ParameterShift} then implies that their gradients with respect to $\theta$ can be efficiently evaluated using our Recurrent Parameter Shift Rule (Algorithm 1). Application of the chain rule then allows us to write $D_A^L(\mathcal{Q}_\theta)$ directly in terms of these quantities. More details are available supplementary information C.

\noindent\inlineheading{Performance Benchmarking} To illustrate the efficacy of quantum sequence models, we benchmark them in the learning of a uniform renewal process~\cite{pyke1961}  (see \fref{fig:renewal}). Renewal processes represent a generalization of Poisson processes. Both involve monitoring a system that, at each time $t$, chooses to emit a tick ($x_t = 1$) or not ($x_t = 0$). Whereas in Poisson processes, the probability of ticking at each time-step is constant, renewal processes generalize this in that the ticking probability can depend on the number of time-steps that have elapsed since the last tick. Such renewal processes appear in diverse time-series settings. Examples include neural spike trains~\cite{crutchfield2015}, device lifetimes~\cite{doob1948}, stochastic clocks~\cite{woods2022}, bus waiting times.

\begin{figure}[htbp]
    \centering
    \includegraphics[width=0.95\linewidth]{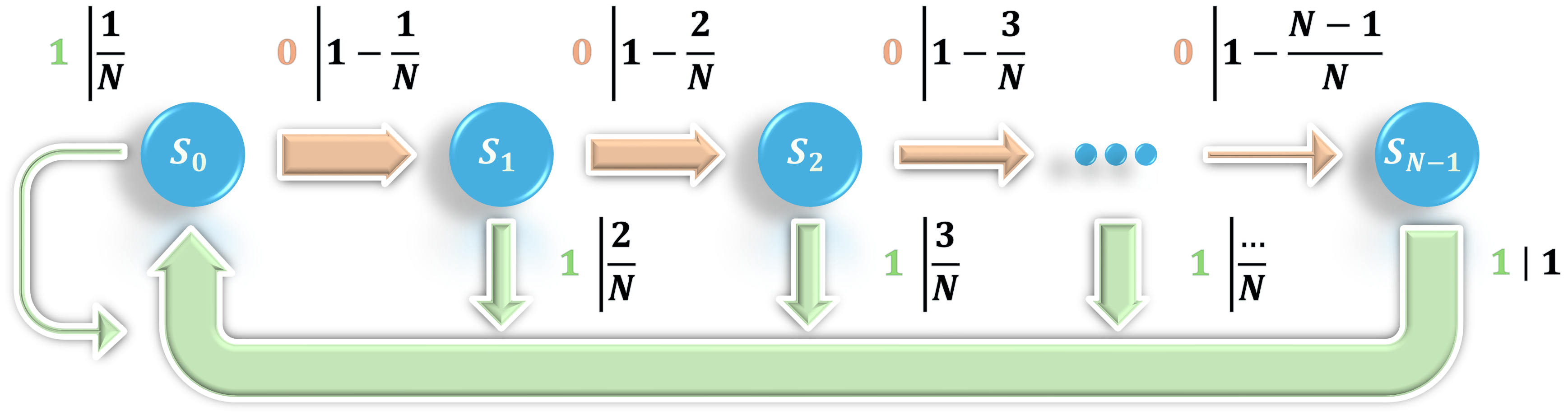}
    \caption{\textbf{Uniform renewal process}. The dynamics of a uniform renewal process $P_N$ or order $N$ can be described as a hidden Markov model (HMM) with $S_0, S_1 \ldots S_{N-1}$ internal states - represented by blue nodes on the directed graph above. The edges between nodes then describe transitions between internal states, where an edge between $S_j$ and $S_k$ with labels
    labeled by $x|p$ represents the statement that a model in state $S_j$ at time-step $t$ will output $x$ and transition to state $S_k$ with probability $p$. For the renewal process above, the HMM describes a binary sequence of mostly $0$s (no-tick) punctuated by $1$s (tick), where the number of $0$s between $1$s is uniformly distributed between $0$ and $N-1$. To enable such statistics, the $HMM$ must have $N$ states, as the probability of the next tick depends on how many time-steps have occurred since the previous tick.}
    \label{fig:renewal}
\end{figure}

Here we consider a family of uniform renewal processes $P_N$ with maximum inter-tick times ranging from $N = 3$ to $N = 8$ states \fref{fig:renewal}. For each $N$, the training data is a sequence of $T$ symbols drawn from the process. We then consider two families of sequence models: (i) a $1$-qubit-memory sequence model with full expressivity and (ii) a $2$-qubit-memory  sequence model built from three layers of a $3$-qubit hardware-efficient ansatz. All models are trained for next-token prediction based on a history of length $8$ ($M = 8$, $L=1$). As a non-recurrent baseline, we consider a variational quantum Born machine consisting of $4$ layers acting on a system of $9$-qubits (See Supplementary materials D for model details). In all scenarios, we assume no pre-knowledge of the underlying process.

Fig. \ref{fig:results}a illustrates the resulting performance with training data size $T = 50,000$, as quantified by the KL-divergence rate. For each model $\mathcal{Q}$, `empirical' refers to the empirical distortion $\hat{D}(\mathcal{Q})$ against available training data. Meanwhile, `true' refers to its performance against the true underlying distribution $P$. We note also that our baseline variational circuit model contains the greatest number of parameters ($140$ for baseline vs $8$ for 1-qubit memory and $33$ for 2-qubit memory), we would generally expect it to outperform recurrent candidates in empirical performance (as it has more degrees of freedom). While this is indeed true against the single-qubit sequence mode, the two-qubit sequence model achieved significantly better performance even on training data. This is somewhat counterintuitive, but it is a consequence that recurrent circuits were much easier to train owing to the mitigation of vanishing gradients (See appendix~\ref{sec:ansatzes} for comparison of gradient amplitudes).

Of course, what really matters is true distortion. Here, the advantage of recurrent models is even more readily apparent. We observe that the true distortion performance of our $2$-qubit-memory model (blue) is well over an order-of-magnitude advantage below that of the variational circuit baseline (orange). In fact, even a single qubit of memory exhibits a notable performance advantage over half the baseline distortion. This extra advantage arises from the reduced generalization gap. Baseline variational circuits always exhibited lower empirical distortion over true distortion, reflecting the usual positive generalization gap. In contrast, our recurrent models reverse this and perform better when measured against the true distribution versus that against training data.

This advantage is further magnified when available training data is scarce. In Fig. \ref{fig:results}b, we illustrate training of the on the $\tau = 5$ process using training data sequences of $500$ and $5000$. In these data-scarce settings, the baseline variational circuit exhibits extreme overfitting. While its empirical performance does not visibly change, the true KL-divergence rates increase by nearly 10-fold from $0.107$ to $1.06$ when training data is reduced from $T = 50000$ to $T = 500$. In contrast, our recurrent circuits exhibit a much lower increase in distortion, from $0.042$ to $0.108$. These results indicate that our recurrent models are significantly more robust to finite statistical fluctuations.

\begin{figure}[htp]
    \centering
    \includegraphics[width=.9\linewidth]{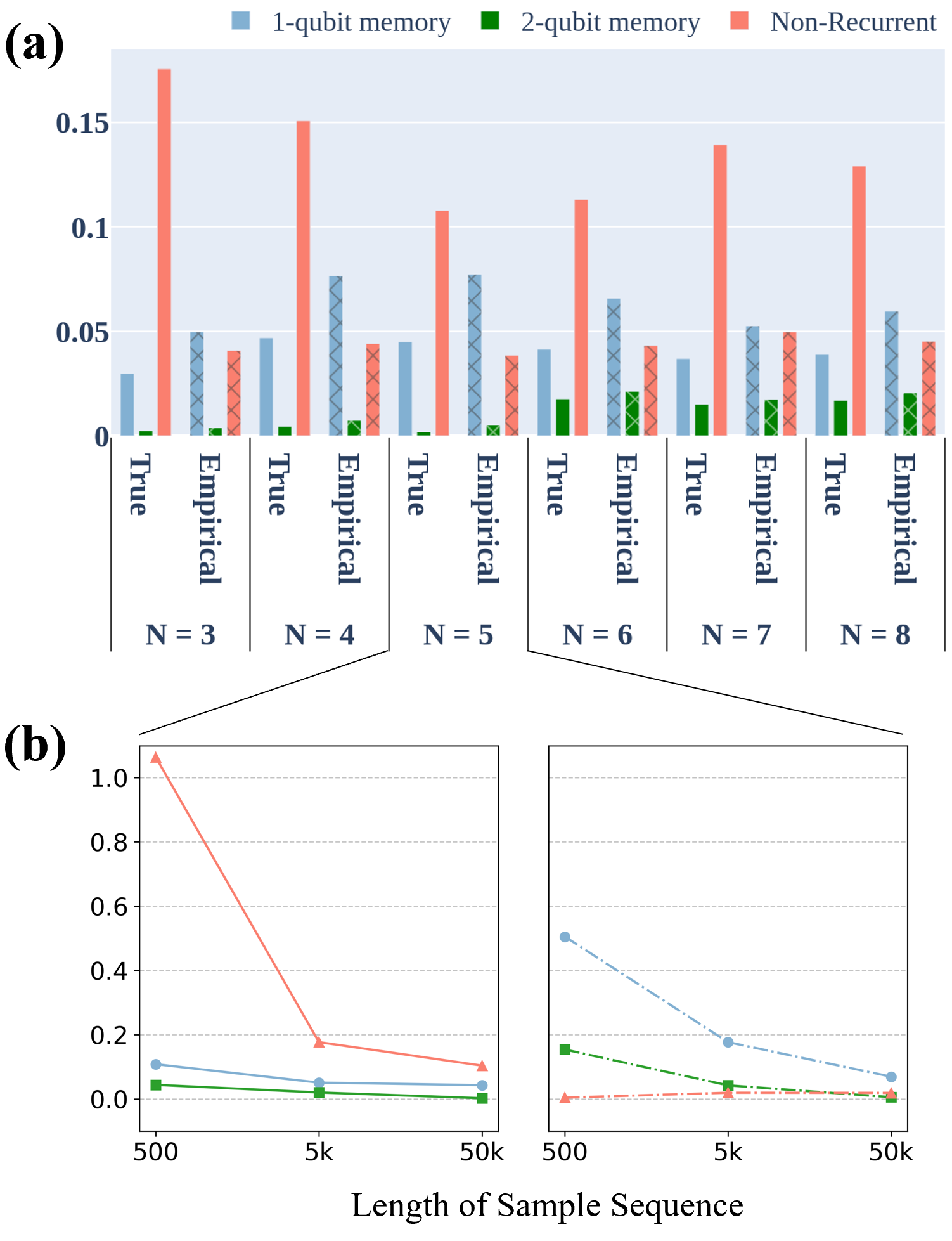}
   \caption{ \label{fig:results} \textbf{Performance Benchmarks}. (a) We compare three quantum models in terms of KL-divergence rates: recurrent quantum circuit with 1-qubit memory (blue), recurrent quantum circuit with 2-qubit memory (green), and a non-recurrent circuit (orange) when tasked to learn uniform renewal processes of order $N$ from $3$ to $8$ when training data size $t = 50,000$. The true distortions are shown as the solid bars, while the empirical distortions are shown as crossed bars. Both 1-qubit and 2-qubit recurrent models (green) clearly outperform the non-recurrent baseline, even when the former does not fit training data as well. (b) The performance advantage is magnified in data-sparse settings where training data (x-axis) is lowered to $5000$ and $500$. While the non-recurrent baseline is able to fit training data almost perfectly (dashed line on right), its true distortion (solid line on left) becomes an order of magnitude higher than its recurrent counterparts.}
    \label{fig:gen_error}
\end{figure}

\noindent\inlineheading{Discussion} Efficient quantum samplers for classical probability distributions are essential for harnessing various iconic quantum-enhanced algorithms for data analytics -- risk analysis~\cite{woerner2019}, quantum Monte Carlo integration~\cite{herbert2022}, quantum signal processing~\cite{blank2021}. Here, we introduce quantum sequence models and associated methods to train them as a way to infer such quantum samplers for stochastic processes directly from observational data. Like classical sequence models, our quantum counterparts leverage a recurrent quantum circuit structure suited for modeling temporal data. In introducing associated cost functions and a recurrent variant of the parameter-shift rule, we show how the usual tools of quantum neural networks can be directly extended to this regime. The resulting output - a quantum superposition of output sequences weighted by their statistical likelihood of occurrence - implies that our quantum sequences have immediate applications in the domains above. We benchmarked our quantum sequence models against baseline variational circuit Born machines, demonstrating significant improvement, especially in regimes with limited training data. When using $2$-qubits of recurrent memory, our sequence models exhibited distortions an order of magnitude below that of the baseline.

Our quantum sequence models complement parallel development in a larger class of recurrent quantum circuits. Quantum recurrent neural networks~\cite{bausch2020,takaki2021} and quantum reservoir computers~\cite{ghosh2019,suzuki2022}, for example, share a similar recurrent structure --- with the primary difference being that they focus on expectation values rather than quantum superposition states as output. While they do not generate quantum samples, their use of time-series prediction has already shown some promises~\cite{chen2024}. Meanwhile, extensions of our models to cases where output registers encode non-trivial states would result in quantum agents~\cite{thompson2017} - which have shown both memory and energetic advantages in executing complex adaptive operations. Further understanding their interrelations could benefit across each of these settings, from new means of adapting existing reservoir computing techniques to the generation of quantum samplers to variational circuits to discover more energy-efficient quantum adaptive agents.

Another natural question is the trainability. Standard variational quantum circuits suffer from vanishing gradients: Any universal ansatz for an $L$-qubit circuit suffers gradients that scale as $e^{-L}$~\cite{mcclean2018, wang2021a}, and thus do not scale for building q-samples of stochastic processes as we seek q-samples over longer time-horizons. In contrast, a recurrent circuit over $L$ time-steps (and thus $O(L)$ qubits) involves $L$ repeated application of the same unitary and thus does not immediately suffer the same roadblocks~\cite{larocca2025}. Instead, scaling costs are likely correlated with the amount of recurrent memory qubits. In computational mechanics, such required memory costs are known as an intrinsic measure of a process's intrinsic complexity~\cite{barnett2015,shalizi2001}. Moreover, this memory measure can be drastically reduced by the use of recurrent quantum models~\cite{elliott2020}. As such, these ideas hint that intrinsic quantum memory benefits in stochastic simulation may also directly impact what processes can be efficiently learned.

\textbf{Acknowledgements} This work is supported by the National Research Foundation of Singapore through the NRF Investigatorship Program (Award No. NRF-NRFI09-0010), the National Quantum Office, hosted in A*STAR, under its Centre for Quantum Technologies Funding Initiative (S24Q2d0009), and the Singapore Ministry of Education Tier 1 Grant RG91/25 and RT4/23. C.Y. is funded by Schmidt Sciences, LLC.

\textbf{Data Availability} The data that support the findings of this study are publicly available at GitHub~\footnote{Canoming/Research-Data. GitHub, \url{https://github.com/Canoming/Research-Data/tree/master/learning_stoch}}.

\bibliographystyle{unsrt}
\bibliography{Stoch}

\appendix

\section{Data collection from long sequences}

In this work, we typically only have access to a single long trajectory produced by a stochastic process.
To justify estimating ensemble quantities from such a trajectory, we recall a standard result of ergodic theory.

Let $(X_t)_{t \ge 0}$ be a stationary stochastic process defined on a probability space $(\Omega,\mathcal{F},\mathbb{P})$, and let $T$ denote the left-shift map $T(\dots,X_0,X_1,X_2,\dots) = (\dots,X_1,X_2,X_3,\dots)$.
The map $T$ is measure-preserving with respect to $\mathbb{P}$ and generates a $\sigma$-algebra $\mathcal{C} \subseteq \mathcal{F}$ of invariant sets, i.e.
\begin{equation}
  \mathcal{C} = \{ C \in \mathcal{F} : T^{-1}(C) = C \} \, .
\end{equation}
The Birkhoff--Khinchin theorem~\cite{muller2022} states that for any uniformly bounded measurable function
$f : \Omega \to \mathbb{R}$ with $\mathbb{E}(\abs{f}) < \infty$, the time average of $f$ along a single trajectory
converges almost surely to the conditional expectation of $f$ with respect to $\mathcal{C}$:
\begin{equation}
    \lim_{L\to\infty} \frac{1}{L} \sum_{t=0}^{L-1} f(X_t)
    = \mathbb{E}\bigl(f \,\vert\, \mathcal{C}\bigr)(X)
    \qquad \text{a.s.}
    \label{eq:birkhoff-appendix}
\end{equation}
Here, $\mathbb{E}(f \vert \mathcal{C})$ is the conditional expectation of $f$ given the invariant $\sigma$-algebra $\mathcal{C}$, and the right-hand side is evaluated along the same realization $X = (X_t)_{t \ge 0}$.

A particularly relevant case for our purposes is when $f$ is an indicator of an event.
For any measurable set $A \in \mathcal{F}$, define $f_A = \mathbb{I}_A$.
Applying~\eqref{eq:birkhoff-appendix} to $f_A$ yields
\begin{equation}
    \lim_{L\to\infty} \frac{1}{L} \sum_{t=0}^{L-1} \mathbb{I}_A(X_t)
    = \mathbb{E}\bigl(\mathbb{I}_A \,\vert\, \mathcal{C}\bigr)(X)
    = \mathbb{P}(A \,\vert\, \mathcal{C})(X)
    \qquad \text{a.s.}
\end{equation}
In other words, the time-averaged frequency of visits to $A$ converges almost surely to the conditional probability $\mathbb{P}(A \vert \mathcal{C})$, which is determined by the initial state of the process.
For cylinder events $A = \{X_{t:t+\ell-1} = x_{0:\ell-1}\}$, this implies that the empirical frequencies of finite words $x_{0:\ell-1}$ converge almost surely to the corresponding conditional probabilities of observing these words, as $L \to \infty$.

If the process is ergodic, the invariant $\sigma$-algebra $\mathcal{C}$ is trivial, so that $\mathbb{P}(A \vert \mathcal{C}) = \mathbb{P}(A)$ almost surely.
In that case, time averages converge to ensemble averages that are independent of the initial state.

\section{Parameter Shift Rule for Recurrent Quantum Circuits}

\subsection{Background of the Parameter Shift Rule}

The parameter shift rule is a method to evaluate the gradient of a function with respect to the parameters of a quantum circuit.
Without the parameter shift rule, the gradient of a function with respect to the parameters of a quantum circuit can only be approximated by the finite difference method, which is inefficient.
Given a function $f(\vb*{\theta})$ of the parameters $\vb*{\theta} = (\theta_1, \theta_2, \dots, \theta_n)$ defined with respect to a quantum circuit $\mathcal{Q}(\vb*{\theta})$, we need to evaluate the gradients $\pdv{\theta_k} f(\vb*{\theta})$ for all $k = 1, 2, \dots, n$ to optimize $\mathcal{Q}$ with gradient descending methods.
For each $k$, we are only interested in one parameter $\theta_k$.
While the $\theta_k$ are independent of each other, we can fix all other parameters $\theta_i$ for $i \neq k$.
Here, we slightly abuse the notation to denote the function as $f(\theta_k)$ by omitting other parameters.

Without the ability to evaluate the gradient directly, we may use the finite difference method to approximate the gradient.
The accuracy of the finite difference approximation
\begin{equation}
    \pdv{\theta_k} f(\theta_k) \approx \frac{f(\theta_k + \delta) - f(\theta_k)}{\delta}
\end{equation}
increases as $\delta$ decreases.
However, a small $\delta$ leads to a vanishing difference $f(\theta_k + \delta) - f(\theta_k)$, which makes the multiplicative error of the evaluation explode.
The parameter shift rule resolves the issue by expressing the derivative of a function as the mean of observable outcomes.
The accuracy of the evaluation is then independent of the choice of $\delta$.

The parameter shift rule is defined for some cost functions in form $f(\theta_k) = \tr[O \mathcal{Q}(\theta_k)\rho \mathcal{Q}^\dagger(\theta_k)]$ (the $\rho$ dependency is implicit in $f$ for tighter expressions), which is the expectation value of an observable $O$ with respect to the output state $\mathcal{Q}(\theta_k)\rho \mathcal{Q}^\dagger(\theta_k)$ of a variational circuit $\mathcal{Q}(\theta_k)$.
In the simplest case where, for each $k\in [1,n]$, there is only one Pauli rotation gate that depends on $\theta_k$, the parameter shift rule reads
\begin{equation}
    \pdv{\theta_k} f(\theta_k) = \frac{1}{2} \qty(f(\theta_k + \frac{\pi}{2}) - f(\theta_k - \frac{\pi}{2})) \ .
\end{equation}

\subsection{Generalization of the Parameter Shift Rule}

However, due to the recurrent structure of our quantum model, the same parameters are shared in every step.
This is slightly different from the conventional parameter shift rule.
To generalize the parameter shift rule to the case where multiple Pauli rotation gates depend on $\theta_k$, we first introduce the following notations.
Assume there are $n$ Pauli rotation gates $U_i$ in a variational circuit that depends on $\theta_k$.
The circuit $\mathcal{Q}$ is first divided into $2n+1$ parts:
\begin{equation}
    \mathcal{Q}(\theta_k) = W_n U_{n-1}(\theta_k) W_{n-1} \cdots W_1 U_0(\theta_k) W_0 \ ,
\end{equation}
where $U_i$ are the Pauli rotation gates, and $W_i$ are the rest of the circuit.
As shown in figure~\ref{fig:CircDecomp}, $W_i$ consists all gates between $U_i$ and $U_{i+1}$.
The choices of $W_i$ are not unique, and this choice does not affect the parameter shift rule as $W_i$ are independent of $\theta_k$.

\begin{figure}[htp]
    \centering
    \includegraphics[width=.45\textwidth]{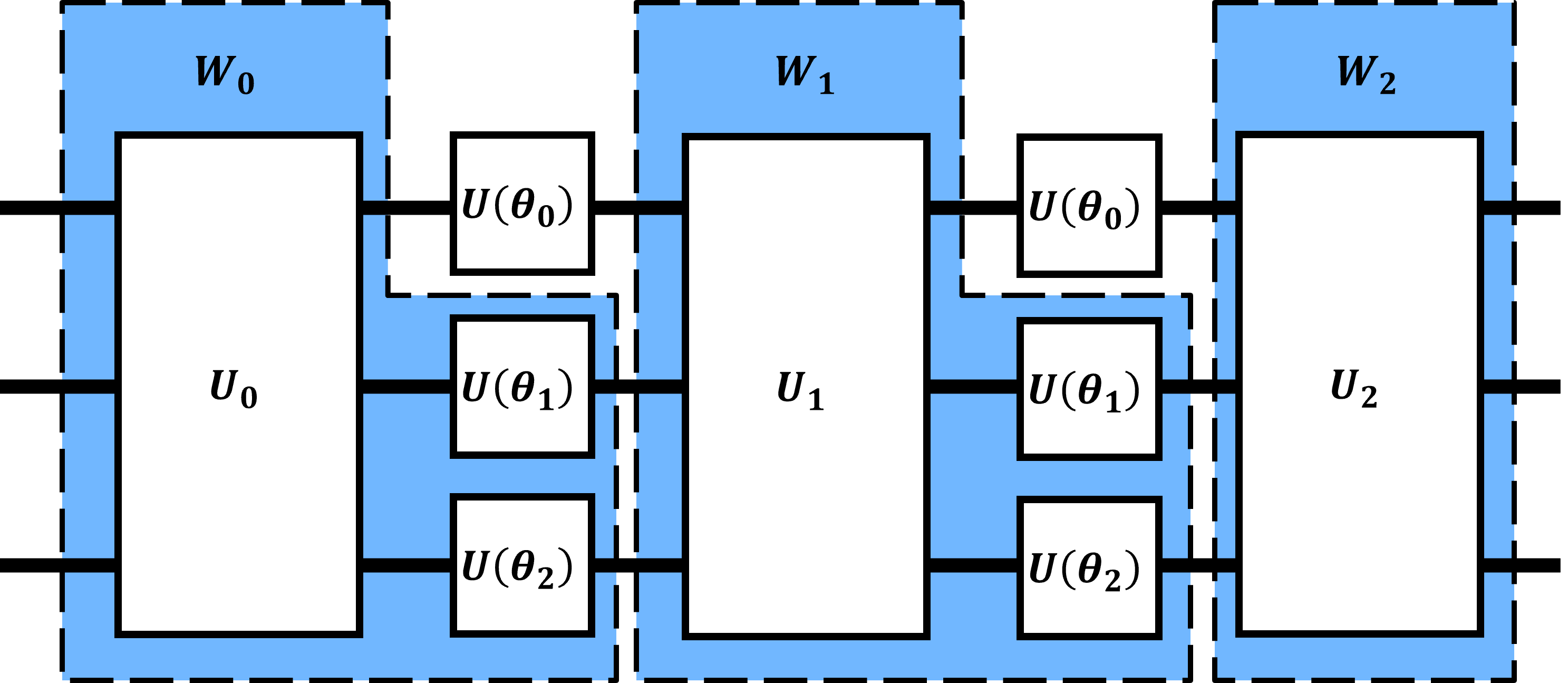}
    \caption{One possible way to decomposition of a variational circuit $\mathcal{Q}(\theta_k)$ into $2n+1$ parts.}
    \label{fig:CircDecomp}
\end{figure}

For a better analysis of the parameter-shift rule, several notations are introduced.
Let $\chU_{\theta_k}$ be the unitary channel that
\begin{equation}
    \chU_{i, \theta_k}(\cdot) := U_i(\theta_k) (\cdot) U_i^\dagger(\theta_k) \ ,
\end{equation}
and let $\chW_i$ be the channels corresponding to $W_i$ for each $i$ respectively.
When $n=2$, according to the product rule of derivatives, the derivative of the composition of multiple channels is
\begin{widetext}
\begin{equation}
    \begin{aligned}
    \pdv{\theta_k} \mathcal{Q} =&\pdv{\theta_k} \qty(\chW_2 \circ \chU_{1,\theta_k} \circ \chW_1 \circ \chU_{0,\theta_k} \circ \chW_0) \\
    = & \chW_2 \circ \qty(\pdv{\theta_k} \chU_{1,\theta_k}) \circ \chW_1 \circ \chU_{0,\theta_k} \circ \chW_0
    +\chW_2 \circ \chU_{1,\theta_k} \circ \chW_1 \circ \qty(\pdv{\theta_k} \chU_{0,\theta_k}) \circ \chW_0 \ .
    \end{aligned}
\end{equation}
\end{widetext}

As each Pauli rotation channel $\chU$ satisfies
$$\pdv{\theta}_k \chU_{i,\theta_k} = \frac{\chU_{i,\theta_k+\pi/2} - \chU_{i,\theta_k-\pi/2}}{2}\ ,$$
the parameter shift rule then reads
\begin{equation}\label{eq:PSR2gates}
    \begin{aligned}
    \pdv{\theta_k} f=&\pdv{\theta_k} \tr[O \mathcal{Q}(\rho)] \\
    =& \tr[O (\chW_2 \circ \qty(\pdv{\theta_k} \chU_{1,\theta_k}) \circ \chW_1 \circ \chU_{0,\theta_k} \circ \chW_0)(\rho)]\\
    &+ \tr[O (\chW_2 \circ \chU_{1,\theta_k} \circ \chW_1 \circ \qty(\pdv{\theta_k} \chU_{0,\theta_k}) \circ \chW_0)(\rho)] \\ 
    =& \tr[O (\chW_2 \circ \chU_{1,\theta_k+\frac{\pi}{2}} \circ \chW_1 \circ \chU_{0,\theta_k} \circ \chW_0)(\rho)]/2 \\
    &- \tr[O (\chW_2 \circ \chU_{1,\theta_k-\frac{\pi}{2}} \circ \chW_1 \circ \chU_{0,\theta_k} \circ \chW_0)(\rho)]/2 \\
    & + \tr[O (\chW_2 \circ \chU_{1,\theta_k} \circ \chW_1 \circ \chU_{0,\theta_k+\frac{\pi}{2}} \circ \chW_0)(\rho)]/2 \\
    &- \tr[O (\chW_2 \circ \chU_{1,\theta_k} \circ \chW_1 \circ \chU_{0,\theta_k-\frac{\pi}{2}} \circ \chW_0)(\rho)]/2 \ .
    \end{aligned}
\end{equation}

To simplify such an expression for general $n$, we further introduce the following notations:
\begin{equation}
    \begin{aligned}
    \mathcal{Q}_{i,\theta_k}^+ := \chW_n \chU_{n-1,\theta_k} \cdots \chW_{i+1} \chU_{i,\theta_k + \pi/2} \chW_{i}\cdots \chU_{0,\theta_k}\chW_0 \ ; \\
    \mathcal{Q}_{i,\theta_k}^- := \chW_n \chU_{n-1,\theta_k} \cdots \chW_{i+1} \chU_{i,\theta_k - \pi/2} \chW_{i}\cdots \chU_{0,\theta_k}\chW_0 \ ; \\
    \end{aligned}
\end{equation}
where $\mathcal{Q}_{i,\theta_k}^+$ is the variational circuit $\mathcal{Q}$ with the $i$-th $\theta_k$ dependent Pauli rotation gate containing being shifted by $\pi/2$, and $\mathcal{Q}_{i,\theta_k}^-$ is $\mathcal{Q}$ with the $i$-th $\theta_k$ dependent Pauli rotation gate being shifted by $-\pi/2$.
Also, let
\begin{equation}
    \begin{aligned}
    f_{i,+1} := \tr[O \mathcal{Q}_{i,\theta_k}^+(\rho)] \ ; \\
    f_{i,-1} := \tr[O \mathcal{Q}_{i,\theta_k}^-(\rho)] \ ; \\
    \end{aligned}
\end{equation}
Then the equation~\eqref{eq:PSR2gates} can be rewritten as
\begin{equation}
    \pdv{\theta_k} f= \frac{f_{0,+1} - f_{0,-1}}{2} + \frac{f_{1,+1} - f_{1,-1}}{2} \ .
\end{equation}

In general, in case there are $n$ Pauli rotation gates that depend on $\theta_k$, the derivative of the variational circuit can be evaluated by
\begin{equation}
    \pdv{\theta_k} f= \sum_{i\in [0,n-1],s\in\qty{-1,+1}} \frac{s\cdot f_{i,s}}{2} \ .
\end{equation}

To see how this value is measured on quantum devices, we need to look into the measurement of the expectation value of the observable $O$.
Let $\mathcal{M}_f = \qty{(\Pi_x, \lambda_x)}$ be a \emph{measurement protocol} corresponding to the function $f$, such that the set $\qty{\Pi_x}$ is a set of operators that forms a POVM measurement, and
\begin{equation}
    \sum_x \lambda_x \tr[\Pi_x \rho] = \tr[O \rho] = f \ ,
\end{equation}
for all quantum states $\rho$.
We may define a \emph{random variable} $Z \in \mathbb{R}$ with the probability distribution $\prob(Z=\lambda\vert \rho) := \sum_{x: \lambda_x = \lambda} \tr[\Pi_x \rho]$.
Then $f$ can be expressed as the expectation value of the random variable that $f = \langle Z\rangle$.
Similarly, each $f_{i,s}$ can be expressed as the expectation value of the random variable $Z_{i,s}$, which is the random variable with respect to the distribution
\begin{equation}
    \prob(Z_{i,s} = \lambda \vert s, i, \rho) := \sum_{x: \lambda_x = \lambda} \tr[\Pi_x \mathcal{Q}_{i,\theta_k}^s(\rho)] \ .
\end{equation}
and
\begin{equation}
    f_{i,s} = \langle Z_{i,s} \rangle = \sum_{\lambda} \lambda \prob(\lambda \vert s,i,\rho) \ .
\end{equation}
Note that this distribution can be evaluated by the same protocol $\mathcal{M}_f$ with the quantum state $\mathcal{Q}_{i,\theta_k}^s(\rho)$, and the number of gates in $\mathcal{Q}_{i,\theta_k}^s$ is same as $\mathcal{Q}$.

In summary,
\begin{equation}
    \begin{aligned}
    \pdv{\theta_k} f & = \sum_{i\in [0,n-1],s\in\qty{-1,+1}} \frac{s\cdot \langle Z_{i,s} \rangle}{2}
    &= \sum_{\lambda, s, i} \frac{\lambda \, s}{2} \prob(\lambda \vert s,i,\rho) \\
    \end{aligned}
\end{equation}
If we i.i.d. sample $i$ and $s$ from the even distribution of the integers $[0,n-1]$ and the set $\qty{-1,+1}$ respectively.
\begin{equation}
    \begin{aligned}
    & \sum_{\lambda, s, i} \lambda \, s \, \prob(\lambda \vert s,i,\rho)\prob(s)\prob(i)
    &= \sum_{\lambda, s, i} \frac{\lambda \, s}{2 n} \prob(\lambda \vert s,i,\rho)
    &= \frac{1}{n} \pdv{\theta_k} f
    \end{aligned}
\end{equation}
It means that, the evaluation of $\pdv{\theta_k} f$ is the mean of a stochastic process
\begin{equation}
    (n s \hat{Z}) \sim \prob(\hat{Z} = \lambda, s, i \vert \rho) \ ,
\end{equation}
where it takes value $n s \lambda$ with probability
\begin{equation}
\begin{aligned}
    \prob(\hat{Z} = \lambda, s, i \vert \rho)  & = \prob(\lambda \vert s,i,\rho)\prob(s)\prob(i) \\
    & = \sum_{x: \lambda_x = \lambda} \prob(x \vert s, i, \rho) \prob(s) \prob(i) \ ,
\end{aligned}
\end{equation}
As a result, we may evaluate the gradient $\pdv{\theta_k} f$ by sampling from the distribution $\prob(\hat{Z} = \lambda, s, i \vert \rho)$.
The sampling procedure is summarized as follows.

\begin{algorithm}[tp]
\caption{Parameter Shift Rule (Repeated Parameters)}
\begin{algorithmic}[1]
    \State \textbf{Input} The protocol $\mathcal{M}_f$ that defines $f$, number of samples $m$;
    \State $n \gets$ the number of occurrences $\theta_k$ in the circuit;
    \State \textbf{Init} $D \gets 0$;
    \MRepeat{$m$ times}
        \State Uniformly sample a number $i \in [0,n-1]$;
        \State Uniformly sample a number $s \in \qty{-1,+1}$;
        \State Sample $x \sim \prob(x\vert s,i,\rho)$ once;
        \State Cumulate $D \gets D + s\lambda_x$;
    \EndRepeat

    \State \textbf{return} The gradient $\pdv{\theta_k}f = \frac{n}{m} D$;
\end{algorithmic}
\end{algorithm}

In the special case where the observable $O$ is diagonal in the computational basis, this method reduces to the algorithm used in the main text.
To see this, we first observe that the value of $\tr[O \rho]$ can be evaluated by the computational basis measurements $\Pi_{x_{-M:L}} = \ketbra{x_{-M:L}}$ along.
Let $Q(x_{-M:L} \vert s,i,\rho) = \tr[\Pi_{x_{-M:L}} \mathcal{Q}_{i,\theta_k}^s(\rho)]$ be the probability distribution of measuring the output $x_{-M:L}$ from the quantum state $\mathcal{Q}_{i,\theta_k}^s(\rho)$ in the computational basis.
Then the random variable $Z_{i,s}$ takes value $\lambda(x_{-M:L})$ with probability $Q(x_{-M:L} \vert s,i,\rho)$.

\begin{figure}
    \centering
    \includegraphics[width=0.48\textwidth]{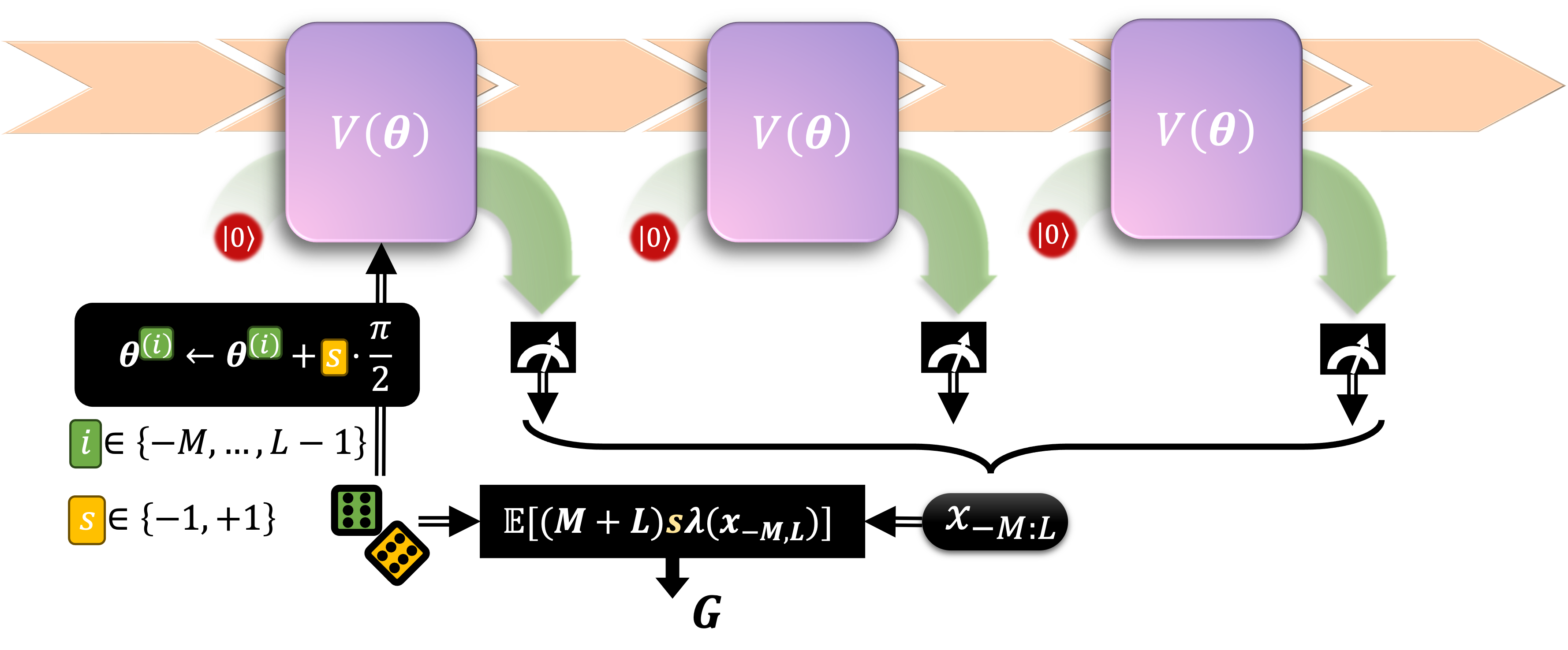}
    \caption[The parameter shift rule for the recurrent quantum circuit.]
    {Given the weights $\lambda(x_{-M:L})$ that define a function $f$, the gradient can be evaluated using the procedure in the algorithm.
    At each iteration, a random step $i$ is selected, and the parameter $\theta$ in the $i$-th unitary $V(\theta)$ is shifted by $s\cdot \frac{\pi}{2}$, where $s$ is randomly chosen from $\qty{-1,+1}$.
    The quantum sample $\ket{\Phi}_{\past{\chi}_M, \future{\chi}_L,S}^{s,i}(\theta)$ is then measured in the computational basis to obtain the output string $x_{-M:L}$.
    The mean value $G = \mathbb{E}[(M+L) s \lambda(x_{-M:L})]$
    is then an unbiased estimator of the gradient $\pdv{\theta} f(\ket{\Phi}_{\past{\chi}_M, \future{\chi}_L,S}(\theta))$.
    }
    \label{fig:ParameterShift}
\end{figure}

\section{Evaluating Cost Functions}

This section describes how to evaluate the co-emission distortion rate $D_A^L$ and its derivatives. Recall that the co-emission distortion rate is defined as
\begin{equation}
    D_A^L(\mathcal{Q})
    = -\frac{1}{L} \sum_{\past{x}_M} P(\past{x}_M)\log \frac{\vb{p}^{\past{x}_M}_L \cdot \vb{q}^{\past{x}_M}_L}
    {\norm{\vb{p}^{\past{x}_M}_L}\norm{\vb{q}^{\past{x}_M}_L}}.
\end{equation}
Let
\begin{equation}
S_{\past{x}_M} = \frac{\vb{p}^{\past{x}_M}_L \cdot \vb{q}^{\past{x}_M}_L} {\norm{\vb{p}^{\past{x}_M}_L}\norm{\vb{q}^{\past{x}_M}_L}}
= \frac{C}{\sqrt{AB}}
\end{equation}
be the \emph{cosine similarity} between the two distributions $\vb{p}^{\past{x}_M}_L$ and $\vb{q}^{\past{x}_M}_L$, and $A=\norm{\vb{p}_L^{\past{x}_M}}^2, B=\norm{\vb{q}_L^{\past{x}_M}}^2$ and $C= \vb{p}_L^{\past{x}_M} \cdot \vb{q}_L^{\past{x}_M}$.
The evaluation of $D_A^L$ then reduces to the evaluation of $S_{\past{x}_M}$ for all $\past{x}_M$.
For simplicity, we may omit the notation of the past in the following discussion, and denote $S = S_{\past{x}_M}$.
The evaluation of $S$ is then reduced to the evaluation of three quantities $A, B$, and $C$ separately.

According to the chain rule, the derivative of $D_A^L$ with respect to $\theta_k$ is the average of the derivatives of $\log S$ for all $\past{x}_M$.
\begin{equation}
    \pdv{\theta_k} \log S =
    \frac{1}{C}\pdv{C}{\theta_k} - \frac{1}{2B} \pdv{B}{\theta_k} \ .
\end{equation}
Note that, as $A$ is independent of $\theta_k$, it is not included in the derivative.
Here we show that the derivatives of $B$ and $C$ can be evaluated efficiently with the help of the parameter shift rule.

By taking the $\lambda(x_{-M:L}) = \vb{q}_L^{\past{x}_M}(x_{0:L-1})$ in the measurement protocol defined in the main text, the derivative of $B$ with respect to $\theta_k$ can be evaluated using the parameter shift rule for recurrent quantum circuits directly.

On the other hand, the quantity $C$ is the inner product of two copies of $\vb{q}_L^{\past{x}_M}$.
It worth to note that $\vb*{q}_L^{\past{x}_M} \cdot \vb*{q}_L^{\past{x}_M} = \prob(\one_{x=x'})$ is the probability of a pair of output strings $(x,x')$ been equal, when drawn from the joint distribution $\vb{q}_L^{\past{x}_M} \times \vb{q}_L^{\past{x}_M}$.
Given two independent processes, each produces a sequence of samples according to $\vb*{q}_L^{\past{x}_M}$, the inner product can be evaluated by counting the portion of events in which the two processes produce the same output.
From this perspective, the quantity $C$ is the expectation value of the observable $O=\sum_x \ketbra{x,x}$ defined on two copies of $\mathcal{Q}$.
Then, by taking $\lambda(x_{-M:L}, x'_{-M:L}) = \one_{x=x'}$, the derivative of $C$ with respect to $\theta_k$ can also be evaluated through extending the parameter shift rule for two copies of unitary~\cite{wang2024} to recurrent quantum circuits.

As all parameterized gates in our variational quantum circuits are Pauli rotations, where $U(\theta_k) = e^{-i\theta_k/2 \sigma_k}$ for some Pauli operator $\sigma_k$, we may apply the parameter shift rule to obtain the derivative of $B$ and $C$ with respect to each of the parameters $\theta_k$.
With the derivatives of $D_A$, we can use the ADAM method to update the parameters of the quantum model that minimize $D_A$.
This training process is summarized in Algorithm~\ref{alg:Variational}.

\begin{algorithm}[tp]
\caption{Training quantum model}\label{alg:Variational}
\begin{algorithmic}[1]
    \State Evaluate $B_{x_{0:k}}$ for all $x_{0:k}$;

    \State Initialize the parameterized quantum model $U(\theta)$;
    \Repeat
        \State Evaluate $A_{x_{0:k}}$ and  $C_{x_{0:k}}$;
        \For{every parameter $\theta_i$}
            \State Evaluate $\pdv{\theta_i} A_{x_{0:k}}$;
            \State Evaluate $\pdv{\theta_i} C_{x_{0:k}}$;
            \State Compute $D_A$ and $\pdv{\theta_i} D_A$ from $A,B,C, \pdv{\theta_i}A$ and $\pdv{\theta_i}B$;
            \State Update $\theta_i$ with ADAM method;
        \EndFor
    \Until {$D_A$ changes less than $\epsilon$ between two epochs;}
    \State \textbf{return} all $\theta_i$;
\end{algorithmic}
\end{algorithm}

\section{Error Analysis}

According to Chebyshev's inequality, for a random variable $X$ with mean
$\mu_X$ and variance $\sigma_X$, the distribution
\begin{equation}
    P(\abs{X-\mu_X}\geq \epsilon \mu_X) \leq \frac{\sigma_X^2}{\epsilon^2\mu_X^2} \ .
\end{equation}
That is, the multiplicative error $\epsilon$ is bounded by the normalized
variance $\frac{\sigma_X^2}{\mu_X^2}$.

Note all of the three quantities $A,B$, and $C$ are defined in the way that $A = \one(P=P^\prime)$, $B = \one(Q=Q^\prime)$, $C = \one(P=Q)$, where $\one$ stands for the indicator function.
As the result, the variance $\sigma_A^2 = \expval{A^2} - \expval{A}^2 = \mu_A - \mu_A^2$, as the indicator functions are idempotent that $\one^2 = \one$.
The same identity can be applied to $B$ and $C$.
In summary, when $X$ takes each of the random variables $A,B,C$,
\begin{equation}
    \frac{\sigma_X^2}{\mu_X^2} = \frac{1-\mu_X}{\mu_X} \ .
\end{equation}
This is small only when $\mu_X$ is non-vanishing.

In order to cope with the multiplicative error, we introduce the logarithm function to our cost.
For an arbitrary function $f(X)$, where $X$ is a random variable, we have
\begin{equation}
    \mu_{f(X)} = f(\mu_{X}) + {\sigma_X^2} \frac{f^{''}(\mu_X)}{2} + \mathcal{O}(\norm{X-\mu_X}^3) \ ,
\end{equation}
Take $f$ to be the logarithm function, it gives
\begin{equation}
    \mu_{\log X} \approx \log \mu_X - \frac{1}{2}\frac{\sigma_X^2}{\mu_X^2} \ .
\end{equation}
Similarly,
\begin{equation}
    \mu_{\log^2 X} \approx \log^2 \mu_X - (1-\log \mu_X)\frac{\sigma_X^2}{\mu_X^2} \ .
\end{equation}
Let $Y:= \log X$ be the random variable of the logarithm of $X$, then
\begin{equation}
    \begin{aligned}
        \frac{\sigma_Y^2}{\mu_Y^2} &= \frac{\mu_{Y^2}}{\mu_Y^2} - 1 \\
        &\approx \frac{\log^2 \mu_X + \log \mu_X \frac{\sigma_X^2}{\mu_X^2} - \frac{\sigma_X^2}{\mu_X^2}}%
        {\log^2 \mu_X + \log \mu_X \frac{\sigma_X^2}{\mu_X^2} - \frac{\sigma_X^4}{4\mu_X^4}} - 1 \\
        &= \frac{\frac{\sigma_X^4}{4\mu_X^4}- \frac{\sigma_X^2}{\mu_X^2}}%
        {\log^2 \mu_X + \log \mu_X \frac{\sigma_X^2}{\mu_X^2} - \frac{\sigma_X^4}{4\mu_X^4}} \\
        &\leq \frac{\frac{\sigma_X^4}{4\mu_X^4}}%
        {\log^2 \mu_X + \log \mu_X \frac{\sigma_X^2}{\mu_X^2} - \frac{\sigma_X^4}{4\mu_X^4}} \\
        &= \qty(\frac{\frac{\sigma_X^2}{2\mu_X^2}}{\log \mu_X - \frac{\sigma_X^2}{2\mu_X^2}})^2 \\
    \end{aligned}
\end{equation}

Now we may substitute $A,B,C$ to $X$, that
\begin{equation}
    \frac{\sigma_Y^2}{\mu_Y^2} \leq \qty(\frac{\frac{1-\mu_X}{\mu_X}}{\log \mu_X - \frac{1-\mu_X}{\mu_X}})^2 
    = \qty(\frac{1-\mu_X}{1-\mu_X - \mu_X \log \mu_X})^2 \ .
\end{equation}
Given $A,B,C$ are all inner product of probability distributions, $\mu_A, \mu_B, \mu_C \in [0,1]$.
Also, the entropy function $-x\log x \in [0,e^{-1}]$ in the range $[0,1]$.
\begin{equation}
    \frac{\sigma_Y^2}{\mu_Y^2} \leq 1 \ .
\end{equation}
Recall that
\begin{equation}
    D_A = \frac{\log A + \log B}{2} - \log C \ .
\end{equation}
According to the error analysis, each of the three terms can be evaluated with the multiplicative error $\epsilon$ with $O\qty(\frac{1}{\epsilon^2})$ samples.
Similarly, as their gradients can be computed with the parameter shift rule, they can be evaluated with the same accuracy.

\section{Ansatzes for Benchmarking}
\label{sec:ansatzes}

To apply our training method, we need to choose an ansatz for the parameterized quantum circuit.
For comparison, in our numerical simulation, we used a fully connected quantum model, which is a hardware-efficient quantum circuit with $4$ layers of full-connected gates as shown in figure~\ref{fig:Full_model}.
Each layer consists of $3$ local rotation gates on each qubit, and a controlled $R_x$ gate between each pair of neighbor qubits.
As the maximum number of qubits in our simulation is $9$, the circuit contains $9\times 4 \times 3 = 108$ rotation gates, and $32$ controlled $R_x$ gates.
That is, the circuit contains $140$ parameters in total.

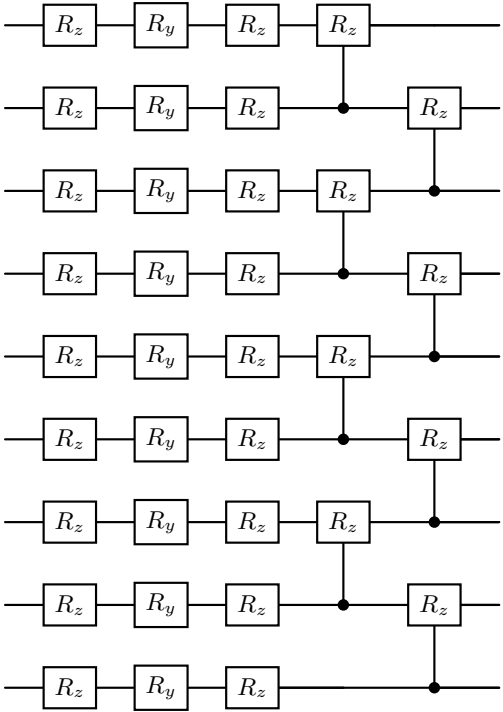
\begin{figure}[htp]
\centering
\resizebox{.8\linewidth}{!}{
\begin{quantikz}
    \qw & \gate{R_z} & \gate{R_y}  & \gate{R_z} & \gate{R_z} & \qw        & \qw  \\
    \qw & \gate{R_z} & \gate{R_y}  & \gate{R_z} & \ctrl{-1}  & \gate{R_z} & \qw  \\
    \qw & \gate{R_z} & \gate{R_y}  & \gate{R_z} & \gate{R_z} &  \ctrl{-1} & \qw  \\
    \qw & \gate{R_z} & \gate{R_y}  & \gate{R_z} & \ctrl{-1}  & \gate{R_z} & \qw  \\
    \qw & \gate{R_z} & \gate{R_y}  & \gate{R_z} & \gate{R_z} &  \ctrl{-1} & \qw  \\
    \qw & \gate{R_z} & \gate{R_y}  & \gate{R_z} & \ctrl{-1}  & \gate{R_z} & \qw  \\
    \qw & \gate{R_z} & \gate{R_y}  & \gate{R_z} & \gate{R_z} &  \ctrl{-1} & \qw  \\
    \qw & \gate{R_z} & \gate{R_y}  & \gate{R_z} & \ctrl{-1}  & \gate{R_z} & \qw  \\
    \qw & \gate{R_z} & \gate{R_y}  & \gate{R_z} & \qw        &  \ctrl{-1} & \qw  \\
\end{quantikz}
}
\caption{\emph{One layer of the fully connected quantum model with 9 qubits.} The model consists of $4$ layers of
such configurations.}
\label{fig:Full_model}
\end{figure}

A similar 3-qubit hardware-efficient quantum circuit is used to construct the 2-qubit memory recurrent quantum circuit for comparison, where the first $2$-qubits are the memory qubits, and the last qubit is the output qubit.
This circuit contains $3$ layers, which counts $3 \times 3 \times 3 = 27$ rotation gates and $6$ controlled $R_x$ gates.
For the 1-qubit memory recurrent quantum circuit, as discussed in the following section, we used a simplified universal quantum circuit containing only $8$ independent parameters.

The number of parameters in each model ($8$ vs. $27$ vs. $140$) significantly affects the landscape of the gradients.
For the $5$-state uniform renewal process, we randomly sampled $100$ initial parameters for each model, and plotted the distribution of the magnitude of the gradient of a random parameter in figure~\ref{fig:gradient}.
As shown in the figure, the gradients of the non-recurrent model (full) are more concentrated at small values, which makes the model harder to train.
In general, a model with a small number of parameters has a higher chance of finding a significant optimization direction.

\begin{figure}[htp]
    \centering
    \includegraphics[width=.8\linewidth]{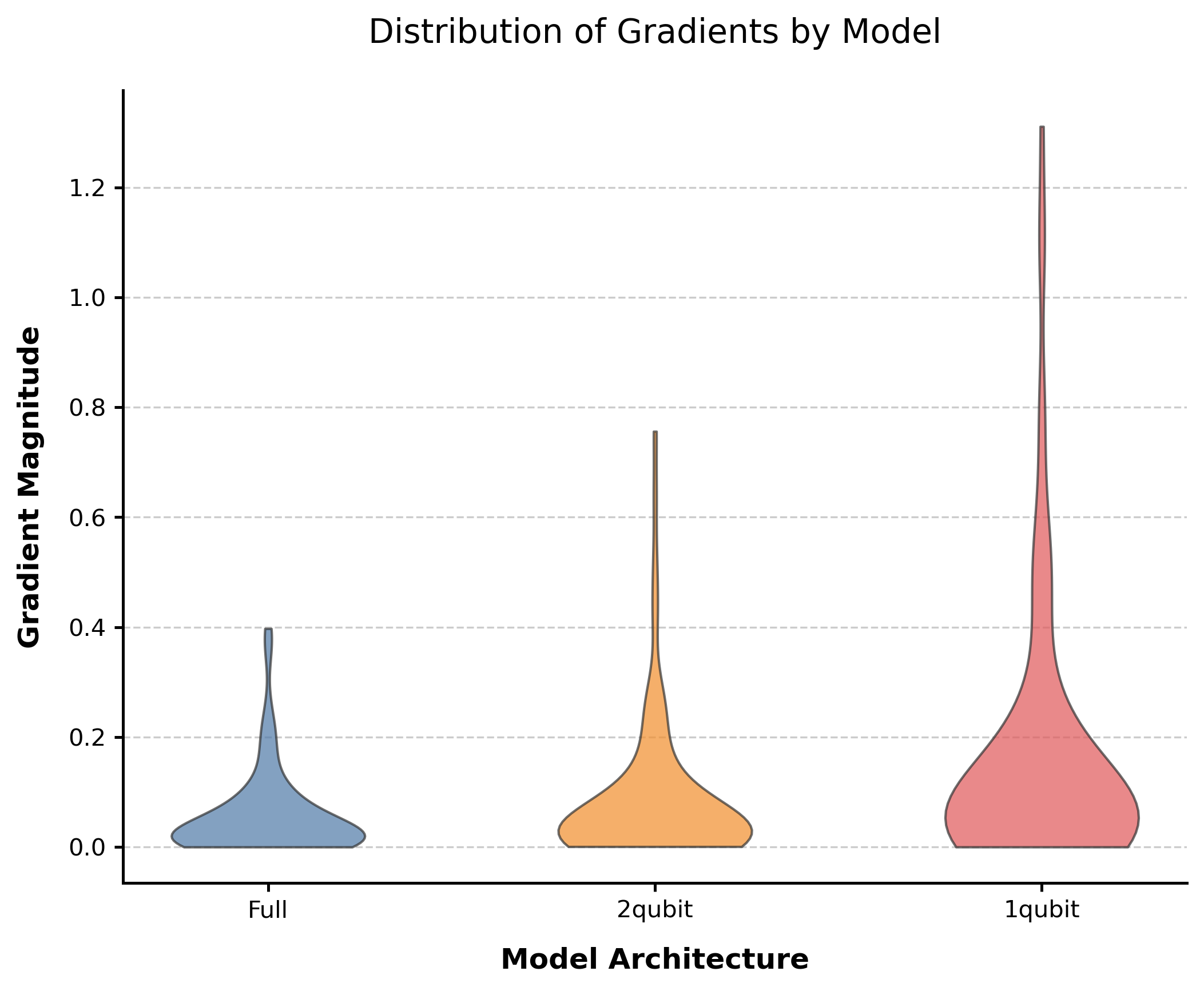}
   \caption{\label{fig:gradient}\textbf{Landscape of the gradients for each model}. Here we randomly sample $100$ initial parameters from each model and plot the magnitude of the gradient of a random parameter. As shown in the figure, the gradients of the non-recurrent model (full) are more concentrated at small values, which makes the model harder to train. In general, a model with a small number of parameters has a higher chance of finding a significant optimization direction.
   }
\end{figure}

\subsection{Simplified Universal Quantum Circuits for Stochastic Process}

In our stochastic model, the initial state of our output state is fixed as $\ket{0}$, and the output is always measured on the computational basis.
Given the constraints, there are many global unitaries that are equivalent in terms of the output distribution.
As a result, we may simplify the universal quantum circuits by eliminating the redundant degrees of freedom induced by such equivalence.

Here, we start from the simplest case, where the output system is a single qubit.

\begin{prop}
Given an $n$-qubit unitary operator $U$, and one of its qubits is measured in the computational basis.
Let $K_0$, $K_1$ be the Kraus operators corresponding to the measurement outcome $0$ and $1$, respectively.
Then there exists two ($n-1$)-qubit unitary operators $U_0$ and $U_1$; and a diagonal real positive matrix $C$, such that the Kraus operators
\begin{equation}
    K_0 = U_0 C; \quad K_1 = U_1 S \ ,
\end{equation}
which simulates the same stochastic process.
Here $S := \sqrt{\mathbb{I}-C^2}$.
These Kraus operators are unique up to a phase gate and a permutation of the computational basis.
\end{prop}

\begin{proof}
At the beginning of each step, we initialize the output system $O$ to $\ket{0}$, and measure it in the computational basis.
The Kraus operators $K_s$ generated by $U$ are
\begin{equation}
    \tilde{K}_s = (\mathbb{I}_M\otimes\bra{0}_O) U (\mathbb{I}_M\otimes\ket{s}_O)
\end{equation}

The projector $(\mathbb{I}_M\otimes\ketbra{s}{s}_O)$ divides the unitary $U$ into two subspaces.
We may divide the matrix of $U$ into 4 parts according to the subspaces:
\begin{equation}
    \mqty(A_{00}&A_{0,1}\\A_{1,0}&A_{1,1}) \ ,
\end{equation}
where $A_{i,j}= (\mathbb{I}_M\otimes\bra{i}_O) U (\mathbb{I}_M\otimes\ket{j}_O)$ is an operator on memory system $M$.
According to the Cosine-Sine (CS) decomposition of unitary operators, $U$ can be represented as
\begin{equation}
    \mqty(\tilde{U}_0 C W_0^\dagger & - \tilde{U}_0 S W_1^\dagger\\
    \tilde{U}_1 S W_0^\dagger & \tilde{U}_1 C W_1^\dagger) \ ,
\end{equation}
where $\tilde{U}_i, W_j$ are unitary operators, and $C^2+S^2 = \mathbb{I}$ are real positive diagonal matrices.
Note that, in each block, this representation is the singular value decomposition (up to a sign) of $A_{i,j}$.
By definition, the Kraus operators are
\begin{equation}
    \tilde{K}_0 = \tilde{U_0} C W_0^\dagger;
    \quad \tilde{K}_1 = \tilde{U_1} S W_0^\dagger \ .
\end{equation}

Since there are no restrictions on the initial state of the memory, we may apply a basis transformation to the initial state of the memory system without affecting the classical output.
Therefore, the canonical form
\begin{equation}
    K_0 = W\tilde{U}_0 C := U_0 C; \quad K_1 = W\tilde{U}_1 S := U_1 S
\end{equation}
simulate the same quantum stochastic process.
The uniqueness (assuming $C$ is not degenerated) of the canonical form can be achieved by restricting the diagonal terms of $C$ be in descending order, and the elements in the first row of $U_0$ are all real (based on uniqueness of singular value decomposition).
\end{proof}

\begin{prop}
All quantum stochastic processes generated by a two-qubit unitary can be
simulated by the circuit~\ref{fig:UniCirc_stoch}.
\end{prop}

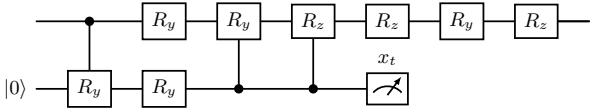
\begin{figure}[htp]
\centering
\resizebox{.45\textwidth}{!}{
\begin{quantikz}
                        &\ctrl{1} & \gate{R_y}        & \gate{R_y} & \gate{R_z}   & \gate{R_z} & \gate{R_y} & \gate{R_z} & \qw \\
    \lstick{$\ket{0}$}  & \gate{R_y}  & \gate{R_y} & \ctrl{-1}        & \ctrl{-1} & \meter{x_t} \\
\end{quantikz}
}
\caption{\label{fig:UniCirc_stoch} \emph{Universal circuit for $2$-output quantum predictive model}.
It contains $8$ parameterized gates.}
\end{figure}

\begin{proof}
This circuit is constructed based on the canonical form of the Kraus operators, which follows~\cite{mottonen2004}.
The circuit of the universal $2\times2$ quantum stochastic process is a controlled $R_Y$ followed by the controlled unitary.
The second part is $U_s \otimes \ketbra{s}{s}_O$ as in the experimental model.
The first part of the circuit is
\begin{equation}
    (\ketbra{0}{0}\otimes R_Y(\theta_0) + \ketbra{1}{1}\otimes R_Y(\theta_0))
    (\mathbb{I}_M\otimes \ketbra{0}{0}_O) \ .
\end{equation}
By taking the product of the two operators, this is
\begin{equation}
    K_s = \mel{s}{R_Y(\theta_0)}{0} U_s \ketbra{0}{0} +
    \mel{s}{R_Y(\theta_1)}{0} U_s \ketbra{1}{1}
\end{equation}

By the definition of $R_Y$, $\mel{0}{R_Y(\theta_0)}{0} = \cos(\theta_0/2)$ and
$\mel{1}{R_Y(\theta_0)}{0} = \sin(\theta_0/2)$.
Therefore
\begin{equation}
    \begin{aligned}
    K_0 = U_0 \mqty(\cos(\theta_0/2)&0\\0&\cos(\theta_1/2)) \\
    K_1 = U_1 \mqty(\sin(\theta_0/2)&0\\0&\sin(\theta_1/2)) \ .
    \end{aligned}
\end{equation}

Since $\theta_0$ and $\theta_1$ are independent, by choosing appropriate $\theta_i$,
$\mqty(\cos(\theta_0/2)&0\\0&\cos(\theta_1/2))$ may simulate all matrix $C$ that
satisfies the CS decomposition.
With the freedom of choosing $U_0$, $U_1$, for every Kraus operator in canonical
form, there exists a circuit in the form of \fref{fig:UniCirc_stoch} that generates
the same Kraus operators.
\end{proof}

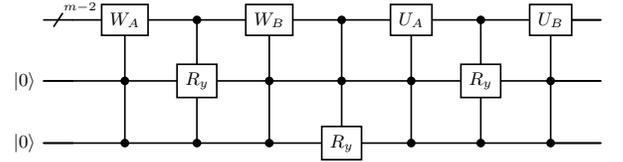
\begin{figure}[htp]
    \centering
    \resizebox{.45\textwidth}{!}{
    \begin{quantikz}
                        & \qwbundle{m-2} & \gate{W_A}  & \ctrl{1}   & \gate{W_B} & \ctrl{1}   & \gate{U_A} & \ctrl{1}   & \gate{U_B} & \qw \\
    \lstick{$\ket{0}$}  & \qw           & \ctrl{-1}       &\gate{R_y}  & \ctrl{-1}      & \ctrl{1}   & \ctrl{-1}      & \gate{R_y} & \ctrl{-1}      & \qw \\
    \lstick{$\ket{0}$}  & \qw           & \ctrl{-1}       & \ctrl{-1}  & \ctrl{-1}      & \gate{R_y} & \ctrl{-1}      & \ctrl{-1}  & \ctrl{-1}      & \qw \\
    \end{quantikz}
    }
    \caption{The decomposition of the canonical form of a 2-qubit-output quantum stochastic process.}
    \label{fig:2out_circ}
\end{figure}

For multi-qubit-output quantum stochastic processes, we may use the same method recursively.
For example, for a 2-qubit-output quantum stochastic process, we may first pick one of the qubits as the output and simulate the process with the method above.
Then the unitary $\tilde{U}_0$, $\tilde{U}_1$, $W_0$, and $W_1$ act on the memory system and the other output qubit.
We may apply the CS decomposition again to $\tilde{U}_0$, $\tilde{U}_1$, $W_0$, and $W_1$ respectively to get the canonical forms, which are decomposed to controlled $R_y$ gates and controlled local unitary gates that act on the memory system as shown in \fref{fig:2out_circ}.
However, although simplified, as this construction is universal for predictive models, it is not surprising that the number of parameters still grows exponentially with the number of output qubits.

\end{document}